\documentclass{article}

\usepackage{microtype}
\usepackage{graphicx}
\usepackage{subcaption}
\usepackage{booktabs} % for professional tables

\usepackage{hyperref}

\usepackage[final]{acl}
\usepackage{times}
\usepackage{latexsym}
\usepackage[T1]{fontenc}
\usepackage[utf8]{inputenc}
\usepackage{microtype}
\usepackage{inconsolata}
\usepackage{graphicx}
\usepackage{subcaption}

\usepackage{booktabs,longtable}

\usepackage{amsmath}
\usepackage{amssymb}
\usepackage{mathtools}
\usepackage{amsthm}
\usepackage{amsmath}
\usepackage{amssymb}
\usepackage{mathtools}
\usepackage{amsthm}
\usepackage{algorithm}
\usepackage{algorithmic}
\usepackage{booktabs}

\usepackage{float}
\theoremstyle{plain}
\newtheorem{theorem}{Theorem}[section]

\newtheorem{lemma}[theorem]{Lemma}
\newtheorem{corollary}[theorem]{Corollary}
\theoremstyle{definition}

\theoremstyle{remark}

\usepackage{amsmath, amssymb, amsthm}
\usepackage{mathtools}
\usepackage{hyperref}
\usepackage{url}

\theoremstyle{plain}

\newcommand{\cX}{\mathcal{X}}
\newcommand{\cY}{\mathcal{Y}}
\newcommand{\cU}{\mathcal{U}}
\newcommand{\cW}{\mathcal{W}}

\usepackage{amsfonts}

\begin{document}
\title{Conformal Privacy Auditing: Calibrated Re-identification Attacks with Statistical Guarantees}

\author{Shuo Huang$^{1}$ \quad Gholamreza Haffari$^{1}$ \quad Xingliang Yuan$^{2}$ \quad Ting Yu$^{3}$ \quad Lizhen Qu$^{1}$\thanks{\;Corresponding author.} \\
  $^{1}$Monash University \quad $^{2}$The University of Melbourne \\
  $^{3}$Mohamed bin Zayed University of Artificial Intelligence \\
  \texttt{\{shuo.huang1,gholamreza.haffari,lizhen.qu\}@monash.edu} \\
  \texttt{xingliang.yuan@unimelb.edu.au} \quad \texttt{ting.yu@mbzuai.ac.ae}}

\maketitle




\begin{abstract}
Empirical identity leakage from released text is increasingly driven by attackers that combine large language models (LLMs) with auxiliary knowledge to link documents to individuals. Existing audits typically report success rates for specific attack pipelines but lack finite-sample statistical guarantees, while training-time protections such as differential privacy are difficult to translate into release-time decisions for individual natural-language documents. We introduce \emph{Conformal Privacy Auditing} (CPA), a distribution-free calibration framework that provides a statistical certificate of re-identification risk for each released document against LLM-empowered adversaries. CPA outputs a \emph{conformal ambiguity set} of candidate identities that is guaranteed to contain the true identity with user-chosen confidence under exchangeability, together with an interpretable leakage proxy derived from set size. CPA supports both logit-access and sampling-only attackers, enabling audits of open-source models and proprietary API models in a unified framework. Across multiple release benchmarks and attacker configurations, CPA achieves calibrated coverage and reveals sharp shifts in certified identifiability as auxiliary knowledge, LLM augmentation, and release mechanisms vary, providing a statistically grounded basis for reporting and comparing release-time linkage risk across attacker configurations, datasets, and release mechanisms alike.
\end{abstract}

\section{Introduction}
\label{sec:intro}

%Privacy issues are of increasing concern alongside the rapid growth of large language models (LLMs) and their applications~\citep{xin2025false}.
%Such applications are often data-intensive. In many areas, including healthcare and law, there is a strong need for organizations to release textual data, such as healthcare notes~\cite{lison2021anonymisation}, legal documents~\cite{deuber2023assessing}, and online posts~\cite{dou-etal-2024-reducing}, for research and product development without breaching the privacy of individuals.
Privacy concerns have intensified alongside the rapid growth of large language models (LLMs) and their widespread applications~\cite{xin2025false}. These applications are inherently data-intensive. In many domains, including healthcare and law, organizations face a strong need to release textual data, such as clinical notes~\cite{lison2021anonymisation}, legal documents~\cite{deuber2023assessing}, and online posts~\cite{dou-etal-2024-reducing}, to support research and product development. At the same time, they must ensure that releasing such data does not compromise the privacy of the individuals it describes. \par

A common practice for releasing sensitive textual data is redaction~\cite{pilan2022tab}, which replaces private information belonging to predefined categories with symbolic placeholders. However, as illustrated in Fig.~\ref{fig:wide_image}, combinations of seemingly innocuous attributes, such as ``37-year-old'', ``female'', ``nephrologist'', and ``Brisbane'', can remain in redacted texts and collectively serve as linkable signals~\cite{xin2025false}, enabling the re-identification of individuals who were intended to be anonymized. Recent studies~\cite{manzanares2024evaluating,staab2023beyond} further show that the widespread use of LLMs has effectively turned them into practical privacy adversaries, capable of extracting salient cues from text and integrating auxiliary information to substantially amplify re-identification risks, even for carefully redacted text.

%For document release, the practical question is no longer “does the text contain a name,” but rather:
%\emph{given plausible side information, how confidently can a modern (LLM-assisted) attacker narrow the document to a small set of candidate individuals or profiles?}

%The existing approaches for \textit{auditing} privacy risks of released texts are limited, especially in the presence of LLMs with strong reasoning capability. Among them, Differential Privacy (DP) provides strong theoretical guarantees by introducing carefully calibrated noise into data~\cite{dwork2006differential}. However, its worst-case guarantees are often overly conservative in practice, leading to substantial utility degradation, while still failing to fully capture the nuanced privacy risks encountered in real-world scenarios. Alternatively, empirical privacy protection methods rely on \textit{empirical} metrics, e.g. top-1 success, to measure the privacy risks. However, such risk evaluation methods do not provide theoretical insights so that there is limited understanding about to what extent the evaluation results may hold on unseen scenarios. 

Existing approaches for \emph{auditing} privacy risks in released textual data remain limited, particularly in the presence of LLMs with strong reasoning capabilities. Differential Privacy (DP) provides formal worst-case guarantees by injecting carefully calibrated noise into the data~\cite{dwork2006differential}. However, these guarantees are often overly conservative in practice, leading to significant utility degradation, while failing to capture the nuanced and context-dependent privacy risks observed in real-world deployments. In contrast, empirical privacy auditing methods quantify privacy leakage using attack-based metrics, such as re-identification success rates; such evaluations provide practical evidence of observed leakage but lack theoretical guarantees, leaving it unclear how well the measured risks generalize to unseen data or stronger attackers.

%Existing approaches fall into two camps, neither of which fully answers the release-decision question. On one hand, \emph{formal} protections such as differential privacy~\cite{dwork2006differential} (DP) offer worst-case guarantees at training time, but are often difficult to interpret for the risk posed by releasing individual documents. On the other hand, \emph{empirical} evaluations frequently rely on ad-hoc attack pipelines and report point estimates (e.g., top-1 success), providing little statistical support for what a curator should conclude from a finite audit set. In high-stakes release settings, it is precisely this lack of calibration that causes harm: a privacy report should be comparable across models and release
%mechanisms, and should remain meaningful when the attacker is a black-box LLM.

To address these challenges, we propose \textbf{Conformal Privacy Auditing} (CPA), a theoretically grounded \emph{black-box} framework for auditing \emph{re-identification} risks of individuals in released texts under a specified attacker model. CPA builds on the theory of Conformal Prediction~\cite{vovk2005algorithmic,bates2021riskcontrol}, a distribution- and model-agnostic framework for uncertainty quantification with formal statistical guarantees. As illustrated in Fig.~\ref{fig:wide_image}, given a released document and attacker background knowledge, CPA employs a score function to identify a \emph{certified ambiguity set} that covers the true target individual with a high probability, e.g., $90\%$, under the chosen attack model. The size of this ambiguity set serves as an interpretable measure of re-identification uncertainty---a large set (e.g., $10^6$ candidates) indicates high uncertainty for the attacker to correctly identify the individual. Importantly, CPA provides \textit{explicit statistical conditions} under which these guarantees hold, enabling principled interpretation and comparison of auditing outcomes across threat models.

\begin{figure*}[t]
    \centering
    \includegraphics[width=0.50\textwidth]{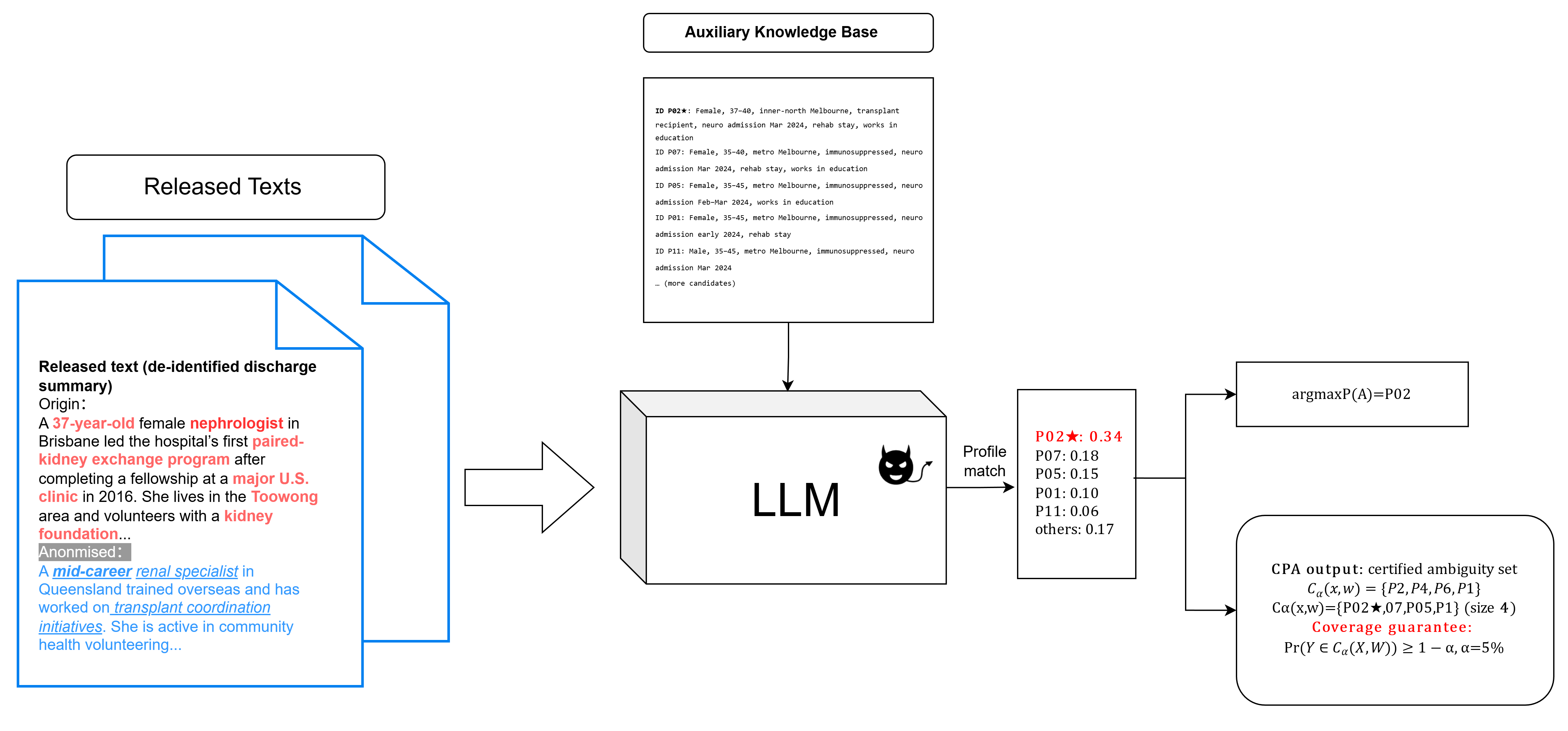}
    \caption{CPA auditing vs.\ top-1 prediction. A conventional attack pipeline (top) outputs a single point prediction; CPA (bottom) calibrates the same attacker into a conformal ambiguity set that contains the true identity with probability at least $1-\alpha$ under the declared threat configuration.}
    \label{fig:wide_image}
\end{figure*}
%\paragraph{Contributions.}
The contributions of this work are three-fold.
\begin{itemize}
  %  \item We formulate the auditing of re-identification risks through identity linkage from released texts as the problem of identifying \emph{conformal ambiguity sets} as privacy certificates using a nonconformity score function. 
   % \item We develop nonconformity score functions specific for LLMs and retrieval models by assuming access to only the log probabilities of candidates or similarity scores respectively, enabling black-box auditing of the corresponding open-source models and API-based models.
   % \item We show that CPA can effectively certify against a \emph{declared family} of attackers and provide statistical guarantee to the released dataset in various adversary settings.
    \item We formulate re-identification auditing of released texts as constructing \emph{conformal ambiguity sets} that serve as privacy certificates, and develop a theoretical framework specifying the statistical conditions under which these certificates provably hold.
    \item We design nonconformity scores for LLM and retrieval-based attackers requiring only candidate log-probabilities or similarity scores, enabling fully \emph{black-box} auditing of open-source models and proprietary API-based models without any internal or logit-level access.
    \item We demonstrate that CPA certifies privacy risks against a \emph{declared family} of attackers, with formal statistical guarantees across diverse release benchmarks, attacker families, and adversarial settings.

\end{itemize}

\section{Preliminaries: Conformal Prediction}
\label{sec:prelim-cp}
Conformal prediction (CP) turns arbitrary model scores into \emph{set-valued} predictions with distribution-free, finite-sample guarantees \citep{vovk2005algorithmic}. It requires no correctly specified probabilistic model; it relies only on \emph{exchangeability}: the joint distribution of $(Z_1,\dots,Z_{n+1})$ with $Z_i=(X_i,Y_i)$ is invariant under permutations of indices (i.i.d.\ sampling is sufficient), so that a test example is statistically indistinguishable from the $n$ examples of a held-out \emph{calibration} sample and ranks of calibration scores form valid quantiles.

Concretely, fix any trained (possibly black-box) predictive system and a \emph{nonconformity score} $s:\mathcal{X}\times\mathcal{Y}\to\mathbb{R}$, where larger values indicate that label $y$ is less compatible with input $x$. Given calibration scores $S_i=s(X_i,Y_i)$, split conformal sets the cutoff $\hat q_\alpha$ to the $\lceil(n+1)(1-\alpha)\rceil$-th smallest calibration score and returns the prediction set $C_\alpha(x)=\{y\in\mathcal{Y}:\ s(x,y)\le \hat q_\alpha\}$. A standard rank argument yields the \emph{finite-sample marginal coverage} guarantee
\begin{equation}
\Pr\!\big(Y_{n+1}\in C_\alpha(X_{n+1})\big)\ge 1-\alpha,
\end{equation}
which is distribution-free and holds for \emph{any} score function $s$ computed consistently on calibration and test examples \citep{vovk2005algorithmic}. In multiclass settings, adaptive prediction sets (APS) construct $s(x,y)$ from the model's ranked probabilities so that sets shrink on easy examples and expand on hard ones while preserving marginal coverage \citep{romano2020classification}.

\paragraph{Relation to threshold tuning on a hold-out set.}
Split CP resembles the familiar practice of tuning a classifier's decision threshold on held-out data, but differs in what is guaranteed. Ad-hoc threshold tuning targets an empirical operating point (e.g., a precision/recall trade-off), makes no finite-sample statement about future data, and inherits any miscalibration of the underlying score. Split CP instead specifies the exact order statistic whose use guarantees marginal coverage at level $1-\alpha$ for the next exchangeable example---for any score function, without parametric or calibration assumptions. In our setting, this is what turns an arbitrary attacker score into a statistically calibrated ambiguity certificate with finite-sample validity guarantees.

\section{Auditing Identity Linkage Risk}
\label{sec:setting}

\subsection{Running example: Release-time identity linkage auditing}
\label{sec:usage}

Consider a healthcare provider releasing de-identified clinical case summaries ($x$) who must ensure that an attacker cannot link them to specific patients using auxiliary background knowledge ($\mathcal{R}$), such as demographics or public registries. The provider uses CPA to simulate a realistic linkage attacker that extracts clues with LLMs and ranks potential matches against a candidate profile database. For each document, the audit produces a \emph{certified ambiguity set} $C_\alpha(x)$: a statistically guaranteed ``short-list'' of suspects containing the true patient with confidence $1-\alpha$. A small set (e.g., $|C_\alpha(x)|=3$) immediately signals that the attacker can narrow the identity to a tiny group, so the provider can withhold or further redact documents whose certified ambiguity falls below a minimum threshold $k$---directly mirroring government guidance that ties release decisions to numeric re-identification risk (e.g., the $0.09$ threshold used by the Information and Privacy Commissioner of Ontario, \citeyear{information2016identification}).
\subsection{Task Formulation}
\paragraph{Released documents and identities.}
We study \emph{identity linkage risk} for released text. Let $X\in\cX$ denote a released document (possibly anonymized or rewritten) and let $Y\in\cU$ denote the true identity of the subject. Here $\cU$ is a \emph{reference population} (e.g., a candidate profile database). The goal of auditing is not to \emph{protect} against an unbounded attacker, but to provide a \emph{quantitative, finite-sample certificate} of linkage risk under a declared attacker capability and budget.

\subsection{Threat Model}
\label{sec:threat-model}
Every CPA certificate is issued \emph{relative to a declared threat configuration}. Following standard practice in security research, we make each component of this configuration explicit below; Table~\ref{tab:knowledge} summarizes the knowledge ladder instantiated in our experiments, and Table~\ref{tab:threat-config} in Appendix~\ref{app:threat-config} maps every experimental suite to its full configuration. A notation table consolidating the symbols used in this and the following sections is provided in Appendix~\ref{app:notation} (Table~\ref{tab:notation}) for reference.

\paragraph{Attacker goal.}
Given a released document $x$, the attacker attempts to link $x$ to the true subject identity $Y$ within the reference population $\cU$.

\paragraph{Side information and tooling as a knowledge level.}
Privacy leakage is attacker-relative: an attacker with access to a detailed candidate profile database and strong search tooling can link text that would be safe under a weaker attacker. We represent attacker side information and tooling by a variable $W\in\cW$, which encodes
(i) which fields are accessible in the candidate profile database (metadata-only vs.\ enriched profiles),
(ii) which auxiliary corpora or snapshots are available,
(iii) whether the attacker may use LLM-based clue extraction or web-search augmentation, and
(iv) decoding/sampling settings.
We interpret different values of $W$ as \emph{knowledge levels} (threat configurations) and report
certificates separately for each knowledge level.
\paragraph{Candidate-pool access (closed vs.\ open world).}
Linkage is meaningful only relative to a finite candidate pool. We assume the attacker operates on a set of candidates
\begin{equation}
\cY(x,w) \subseteq \cU,\qquad |\cY(x,w)|<\infty,
\end{equation}
derived from an attacker-accessible candidate profile database, demographic filters, or a fixed snapshot corpus. In the \emph{closed-world} setting the true identity is guaranteed to lie in $\cY(x,w)$; in the \emph{open-world} setting it may be absent. When we restrict scoring to a top-$K$ retrieved subset for computational reasons, this truncation is an engineering choice that can re-introduce open-world behavior; we quantify the resulting effect on coverage empirically in the stress tests of Section~\ref{sec:results-robustness}.

\paragraph{Query interface and budget.}
Given $(x,w)$ and candidate pool $\cY(x,w)$, the attacker pipeline produces a ranking or
distribution over candidates:
\begin{equation}
p_A(\cdot \mid x,w) \in \Delta(\cY(x,w)).
\end{equation}
Our framework treats the attacker as a black box. In the weakest interface we assume only \emph{sampling-only access} via repeated forced-choice queries over $\cY(x,w)$, with a declared per-document sampling budget $m$; when candidate log-probabilities are exposed we use them directly. This covers retrieval-based linkers (e.g., similarity search over profiles) and LLM-based pipelines that infer attributes/clues and match them to profiles, including proprietary API models.

\paragraph{Scope of the certificate.}
The candidate-pool construction, release mechanism, query mode, knowledge level $w$, sampling budget $m$, and the attacker's randomness model are all part of the declared configuration, and CPA certifies coverage \emph{for that configuration only}: any change to the release distribution or attacker pipeline (e.g., a new prompt, model, or web snapshot) requires recalibration. We stress-test this assumption boundary empirically in the robustness analysis of Section~\ref{sec:results-robustness} and the full tables of Appendix~\ref{app:robustness}.

\paragraph{Auditing goal and outputs.}
For each document $(x,w)$, the auditor outputs a \emph{certified ambiguity set} $C_\alpha(x,w)\subseteq \cY(x,w)$ such that the true identity is included with user-chosen confidence $1-\alpha$. The set size $|C_\alpha(x,w)|$ is the \emph{primary certificate}: it counts how many candidate identities remain statistically plausible after the attack. For reporting convenience we also summarize leakage by the \emph{inverse ambiguity proxy}
\begin{equation}
L(x,w) := \frac{1}{\max\{1,|C_\alpha(x,w)|\}}.
\end{equation}
We emphasize that $L(x,w)$ is \emph{not} a probability of re-identification: it is a monotone transform of the certified set size, and because it decays rapidly with $|C_\alpha|$, seemingly small values can still represent substantial risk (e.g., $L=0.1$ means the attacker has certifiably narrowed the subject to $10$ candidates, which may be unacceptable under many release policies). $L(x,w)$ should be interpreted as a \emph{certificate of attacker ambiguity} rather than a cryptographic privacy guarantee.

\paragraph{Candidate-miss (open-world) caveat.}
In open-world settings the true identity may be absent from the candidate pool. We measure this via the candidate-miss probability
\begin{equation}
\rho(w) := \Pr\big(Y\notin \cY(X,w)\big),
\end{equation}
reported empirically as \texttt{candidate\_miss\_rate}. Our coverage guarantees apply conditional on $Y\in \cY(X,w)$ and degrade gracefully with $\rho(w)$, as quantified by the open-world bound in Corollary~\ref{cor:open-world}.

\section{Methodology: Conformal Privacy Auditing (CPA)}
\label{sec:method}

\subsection{Nonconformity for Identity Linkage: APS Mass Score}
\label{sec:score}

We require a scalar score that measures how much probability mass must be accumulated to include the true identity under attacker $A$. Let $\pi_1(x,w),\pi_2(x,w),\dots$ be candidates sorted by decreasing $p_A(\cdot\mid x,w)$. We use an APS-style \emph{mass score}, restricted to the
candidate pool:

\begin{equation}
s(x,y;w) := \sum_{\substack{j:\,p_A(\pi_j\mid x,w) \\ \ge p_A(y\mid x,w)}} p_A(\pi_j\mid x,w) \in (0,1],
\label{eq:aps}
\end{equation}
where we abbreviate $\pi_j\equiv\pi_j(x,w)$.
Smaller $s(x,y;w)$ indicates that the attacker concentrates probability on the truth (high linkage
risk). Larger $s(x,y;w)$ indicates diffuse belief (lower linkage risk).

\paragraph{CPA is not tied to a single score.}
The calibration layer of CPA requires only an exchangeable scalar score; the APS mass score is one instantiation. In particular, the weights in \eqref{eq:aps} need not be calibrated probabilities: for retrieval attackers they are pseudo-posteriors obtained by softmax-normalizing similarity scores, and for ranking-only attackers a purely rank-based nonconformity score is equally valid. Our ablations (Section~\ref{sec:results-robustness}, Table~\ref{tab:score-ablation}) show that APS, rank-based, and raw-probability constructors all achieve valid coverage; we adopt APS as a robust default rather than claiming it is uniformly dominant.

\paragraph{Sampling-only interface.}
When logits are unavailable, we query the attacker $m$ times with forced-choice decoding over $\cY(x,w)$ and form empirical frequencies $\hat p_m(\cdot\mid x,w)$. We then compute the plug-in score $\hat s_m(x,y;w)$ by replacing $p_A$ with $\hat p_m$ in \eqref{eq:aps}. This yields a fully
black-box conformal pipeline applicable to proprietary LLM APIs that expose no logits.

\subsection{Split Conformal Ambiguity Sets}
\label{sec:split-cpa}

Given calibration data $\{(X_i,Y_i,W_i)\}_{i=1}^n$ and a fixed knowledge level $w$ (or conditioning on $W=w$), we compute scalar scores
\begin{equation}
T_i :=
\begin{cases}
s(X_i,Y_i;w), & \text{probability/logits available},\\
\hat s_m(X_i,Y_i;w), & \text{sampling-only}.
\end{cases}
\end{equation}
Let $T_{(1)}\le \cdots \le T_{(n)}$ be the sorted scores and define the split-conformal threshold
\begin{equation}
k := \left\lceil (n+1)(1-\alpha)\right\rceil,\qquad \hat q_\alpha := T_{(k)}.
\label{eq:qhat}
\end{equation}
For a new document $(x,w)$, we output the calibrated ambiguity set, i.e., the prediction set that contains the true identity with high confidence:
\begin{equation}
C_\alpha(x,w)
:= \{y\in \cY(x,w): s(x,y;w)\le \hat q_\alpha\},
\label{eq:set}
\end{equation}
using $s$ or $\hat s_m$ depending on access. Note that the inequality direction ``$\le$'' in \eqref{eq:set} follows the \emph{nonconformity} convention: top-ranked candidates have \emph{smaller} cumulative mass scores, so the set retains all candidates that the attacker cannot statistically rule out. In particular, $\hat q_\alpha$ is a calibrated quantile of nonconformity scores, not a probability-of-linkage threshold. Section~\ref{sec:theory} gives the finite-sample coverage guarantee under exchangeability, together with its open-world extension.

\subsection{Instantiations of the Attacker Pipeline}
\label{sec:instantiations}

Our framework is attacker-agnostic: any pipeline that induces a distribution over candidates can be conformalized without any modification.

\paragraph{Retrieval-CPA (profile matching).}
We construct a profile text representation for each candidate profile (using a controlled schema) and define $p_A(\cdot\mid x,w)$ from normalized similarity scores (e.g., TF--IDF cosine similarity). Knowledge level $w$ determines which profile fields and which query transformations are permitted (e.g., metadata-only vs.\ keyword profiles; with or without LLM clue extraction at query time).

\paragraph{Profile-CPA (LLM extraction $\rightarrow$ matching).}
An LLM attacker extracts attributes from $x$ (e.g., occupation, location, age range), producing distributions over attribute values via sampling. These distributions induce probabilistic match scores against candidate profiles, yielding $p_A(\cdot\mid x,w)$ over identities.
We then apply CPA exactly as in \eqref{eq:qhat}--\eqref{eq:set}. This instantiation matches a realistic threat model in which LLMs infer quasi-identifiers and link to a candidate profile. The full procedure is listed in Algorithm~\ref{alg:cpa} in Appendix~\ref{app:alg1}.

\section{Privacy guarantee for CPA}
\label{sec:theory}

This section formalizes the statistical guarantees provided by Conformal Privacy Auditing (CPA).
Our results are distribution-free (no parametric assumptions on data) and apply to black-box
attackers, provided an exchangeability assumption holds for the calibration and evaluation
distributions under the declared attacker knowledge level $w$ (Section~\ref{sec:threat-model}).

\subsection{Assumptions}
\label{sec:theory-assumptions}

\paragraph{Exchangeability (within a declared knowledge level).}
Fix a knowledge level $w\in\cW$ that specifies the attacker's side information and tooling (Section~\ref{sec:threat-model}).
We assume that the sequence $\{(X_i,Y_i)\}_{i=1}^{n+1}$ used for calibration/evaluation is
exchangeable under this fixed $w$ (i.i.d.\ sampling is sufficient).
Equivalently, in the most general form we assume $\{(X_i,Y_i,W_i)\}_{i=1}^{n+1}$ is exchangeable;
when experiments are run stratified by category and knowledge level (as in our suites),
$W_i\equiv w$ is constant and exchangeability reduces to exchangeability of $(X_i,Y_i)$ within that stratum.

\paragraph{Attacker pipeline and randomness.}
For each example $i$, the attacker pipeline (retrieval or LLM-based) may involve internal
randomness---e.g., sampling-only forced-choice decoding, stochastic clue extraction, or stochastic
rewriting. We denote this randomness by $U_i$ and assume:
(i) $U_i$ are i.i.d.\ across examples, and
(ii) $U_i$ are independent of the data $(X_i,Y_i,W_i)$.
This reflects operational practice: each audit query uses an independent random seed / sampling
stream, and the auditor does not adversarially choose randomness based on the content of any individual example.

\paragraph{Candidate-pool miss probability.}
In open-world linkage, the true identity may be absent from the candidate pool $\cY(X,w)$.
We measure this using the miss probability
\begin{equation}
\Pr\big(Y\notin \cY(X,w)\big)\le \rho(w),
\label{eq:rho-theory}
\end{equation}
reported empirically as \texttt{candidate\_miss\_rate}. Our strongest guarantees apply conditional on
$Y\in\cY(X,w)$ and degrade gracefully with the miss probability $\rho(w)$, as shown in Corollary~\ref{cor:open-world}.

\subsection{CPA Validity for a Fixed Attacker Configuration}
\label{sec:theory-cpa}

Recall the APS-style nonconformity score (restricted to $\cY(x,w)$):
\[
s(x,y;w)
= \sum_{\substack{y':\, p_A(y'\mid x,w)\\ \ge p_A(y\mid x,w)}} p_A(y'\mid x,w)\in(0,1].
\]
In the sampling-only interface, $p_A$ is replaced by empirical frequencies $\hat p_m$, yielding
$\hat s_m(x,y;w)$. CPA uses split conformal calibration: compute scores on a calibration set,
take $\hat q_\alpha$ as the $\lceil(n+1)(1-\alpha)\rceil$-th order statistic, and return
$C_\alpha(x,w)=\{y: s(x,y;w)\le \hat q_\alpha\}$.

\begin{lemma}[Score exchangeability under sampling-only access]
\label{lem:score-exch}
Let $\{(X_i,Y_i,W_i)\}_{i=1}^{n+1}$ be exchangeable and let $U_i$ be i.i.d.\ independent randomness
used by the attacker pipeline to compute $\hat s_m(X_i,Y_i;W_i)$.
Define scalar scores
\[
T_i := \hat s_m(X_i,Y_i;W_i) = f(X_i,Y_i,W_i,U_i)
\]
for a deterministic measurable function $f$ (deterministic given the example and the sampled outputs).
Then $\{T_i\}_{i=1}^{n+1}$ is an exchangeable sequence of scalar scores.
\end{lemma}

\begin{proof}[Proof sketch]
Augment each example to $(X_i,Y_i,W_i,U_i)$. Exchangeability of $(X_i,Y_i,W_i)$ and i.i.d.\ independence
of $U_i$ implies the augmented sequence is exchangeable. Since $T_i$ is a deterministic function of
the augmented example, exchangeability is preserved under deterministic measurable mappings of the entire augmented example sequence.
\end{proof}

\begin{theorem}[Finite-sample CPA validity for identity linkage]
\label{thm:cpa-valid}
Fix a knowledge level $w$ and attacker pipeline $A$. Assume:
(i) exchangeability holds for calibration and test examples under $w$,
(ii) $Y\in \cY(X,w)$ almost surely (closed-world or $\rho(w)=0$), and
(iii) sampling randomness (if used) satisfies Lemma~\ref{lem:score-exch}.
Let $\hat q_\alpha$ be the split-conformal threshold computed from $n$ calibration scores at level $\alpha$,
and let $C_\alpha(\cdot,w)$ be defined by the rule $s(x,y;w)\le \hat q_\alpha$ (or $\hat s_m$ in the sampling-only case).
Then for a new test example $(X_{n+1},Y_{n+1})$,
\begin{equation}
\Pr\big(Y_{n+1}\in C_\alpha(X_{n+1},w)\big) \ge 1-\alpha.
\end{equation}
\end{theorem}

\begin{proof}[Proof sketch]
Let $T_i$ denote the scalar scores used for conformalization. By exchangeability (Lemma~\ref{lem:score-exch} if needed),
the rank of $T_{n+1}$ among $\{T_1,\dots,T_n,T_{n+1}\}$ is uniform. With $\hat q_\alpha$ chosen as the
$k$-th order statistic for $k=\lceil(n+1)(1-\alpha)\rceil$, we obtain
$\Pr(T_{n+1}\le \hat q_\alpha)\ge 1-\alpha$. Membership $Y\in C_\alpha(X,w)$ is equivalent to $T\le \hat q_\alpha$.
\end{proof}

\begin{corollary}[Open-world validity with candidate-miss probability]
\label{cor:open-world}
Fix $w$ and suppose $\Pr(Y\in \cY(X,w))\ge 1-\rho(w)$.
If conditional on the event $E=\{Y\in \cY(X,w)\}$ the CPA set satisfies
$\Pr(Y\in C_\alpha(X,w)\mid E)\ge 1-\alpha$, then
\begin{equation}
\begin{aligned}
\Pr\big(Y\in C_\alpha(X,w)\big)&\ge (1-\rho(w))(1-\alpha)\\
&\ge 1-(\rho(w)+\alpha).
\end{aligned}
\end{equation}
\end{corollary}

\begin{proof}
By the law of total probability,
$\Pr(Y\in C_\alpha)=\Pr(E)\Pr(Y\in C_\alpha\mid E)$, and apply the two bounds.
\end{proof}

Theorem~\ref{thm:cpa-valid} states that CPA outputs a \emph{valid ambiguity set} containing the truth
with probability at least $1-\alpha$ under the declared attacker knowledge level $w$.
The ambiguity proxy $1/|C_\alpha|$ is therefore a calibrated \emph{certificate of attacker ambiguity},
not a claim about universal privacy against unbounded attackers.

\section{Experiments}
\label{sec:exp}

\paragraph{Goal.}
We evaluate \emph{conformal privacy auditing} as a reporting layer for document release.
Given a released text $x$ and an attacker-side \emph{candidate profile database} $\mathcal{R}$ that encodes auxiliary knowledge,
our auditor outputs a \emph{certified ambiguity set} $C_\alpha(x)\subseteq \mathcal{R}$.
Intuitively, $|C_\alpha(x)|$ quantifies how many plausible identities remain after the attack pipeline is applied, while the
finite-sample conformal guarantee ensures calibrated uncertainty under exchangeability.

\subsection{Datasets and Release Mechanisms}
\label{sec:exp:data}

We use datasets where each released document is associated with a ground-truth identity/profile, enabling identity linkage evaluation.
For each dataset we construct (i) released documents $x$ (e.g., anonymized text) and (ii) a candidate profile database $\mathcal{R}$ of candidate profiles
built from \emph{non-released} sources (e.g., original text, metadata, or historical documents), reflecting attacker side information.
Because certified ambiguity is only meaningful relative to the candidate pool, we describe the per-dataset candidate-pool construction explicitly below; Table~\ref{tab:threat-config} in Appendix~\ref{app:threat-config} maps every experimental suite to its full threat configuration.

\textbf{TextWash.}
TextWash~\cite{kleinberg2022textwash} provides paired \texttt{orig} and \texttt{anon} person descriptions.
We treat \texttt{anon} as the released document and build the candidate profiles from \texttt{orig}.
The benchmark includes three categories (famous, semi-famous, fictional), which we interpret primarily as differing \emph{availability of external context}:
external web knowledge is abundant for famous individuals, limited for semi-famous, and absent for fictional entities.
\emph{Candidate pool:} one profile per identity in the same category, constructed from the \texttt{orig} description (with identity tokens removed); linkage is evaluated within-category, so the pool size equals the size of the category ($\approx 400$ candidate profiles).

\textbf{TAB (text anonymization benchmark).}
TAB~\cite{pilan2022tab} consists of text records where anonymization may preserve substantial lexical content.
We use TAB as an easy leakage regime to verify that the auditor reports near-singleton ambiguity sets when linkage is straightforward.
\emph{Candidate pool:} all $1{,}268$ case records, each represented by a profile built from metadata only (K1), extracted keywords (K2), or case text (K3), with no verbatim copying of the released document into the profiles.

\textbf{WikiBio.}
WikiBio~\cite{stranisci2023wikibio} provides biography-style text with structured fields.
We use WikiBio to control \emph{profile completeness} by varying which fields are exposed to the attacker.
\emph{Candidate pool:} $1{,}000$ candidate profiles built from infobox-style fields or extracted keywords, audited against $500$ released biographies drawn from the prepared slice.

\textbf{Blog Authorship Corpus.}
We use the Blog Authorship corpus~\cite{schler2006effects} as a large-scale linkage setting.
To keep experiments tractable, we subsample to a fixed number of authors and released posts.
\emph{Candidate pool:} one profile per author, built by concatenating $W$ held-out posts ($W\in\{1,3,10\}$) that are disjoint from the released posts, so that $W$ directly controls the strength of attacker side information in this suite.

\paragraph{Candidate-pool truncation and miss rate.}
For expensive attacker pipelines we optionally restrict conformal scoring to the top-$K$ retrieved candidates---an \emph{engineering choice} that can violate the closed-world condition of Theorem~\ref{thm:cpa-valid} if the true identity is dropped. We therefore report the candidate-miss rate with both \emph{conditional} coverage (given the truth is retained) and \emph{unconditional} coverage, and quantify the effect of the truncation level $K$ in the stress tests of Section~\ref{sec:results-robustness}.

\subsection{Attack Pipelines and Knowledge Levels}
\label{sec:exp:attack}

Our threat model is a \emph{declared attacker pipeline} that maps a released document to a ranked list or distribution over candidates, given side information such as locations, gender, and occupations (Section~\ref{sec:threat-model}).
We study two complementary families of attackers.

\paragraph{Retrieval attacker.}
We build a text index over profiles.
Given a released document $x$, the attacker forms a query $q(x)$ and retrieves a ranked list of candidates.
We instantiate a \emph{knowledge ladder} by varying the information available in profiles and the query construction:
(i) a minimal schema subset (weak),
(ii) a richer profile description including auxiliary facts (medium),
and (iii) query augmentation using LLM-extracted clues from $x$ (strong).
For TextWash famous/semi-famous, we additionally consider a web-assisted attacker that can enrich the query $q(x)$ with web search results.

\paragraph{LLM attribute/clue attacker (sampling-based).}
We also evaluate an LLM that extracts attributes or salient clues from $x$ via repeated sampling.
These predictions are matched against the profiles to produce candidate scores.
We treat the number of LLM samples $m$ as a controllable attacker resource and study its impact on the resulting audit outcomes.
% \begin{table*}[ht]
% \centering
% \caption{Retrieval-CPA on TAB under three attacker side-information settings.}
% \label{tab:tab_suite}
% \begin{tabular}{lrrrrrr}
% \toprule
% Attacker setting & Top-1 & Coverage & Median $|C_\alpha|$ & Mean $|C_\alpha|$ & $\mathbb{E}[1/|C_\alpha|]$ & $\hat q_\alpha$ \\
% \midrule
% K1 meta-only & 0.008 & 1.000 & 1268 & 1268.00 & 0.000789 & 1.00 \\
% K2 keywords & 0.923 & 0.965 & 3 & 2.66 & 0.39 & 0.00304 \\
% K3 keywords+LLM clues & 0.923 & 0.975 & 3 & 2.98 & 0.336 & 0.0035 \\
% \bottomrule
% \end{tabular}
% \end{table*}
\begin{table*}[ht]
\centering
\small
\begin{tabular}{lrrrrr}
\toprule
Attacker setting & Top-1 & Cov. & Med. $|C|$ & Mean $|C|$ & Mean $1/|C|$ \\
\midrule
\multicolumn{6}{l}{\textbf{TAB-Direct}} \\
\midrule
K1 meta / direct & 0.052 & 0.929 & 1268 & 1179.5 & 9.57e-04 \\
K2 keywords / direct & 0.973 & 0.973 & 1 & 1 & 1.0000 \\
K3 case-text / direct & 0.994 & 0.994 & 1 & 1 & 1.0000 \\
K2 keywords / LLM-clues & 0.715 & 0.970 & 20 & 20.1 & 0.0499 \\
\midrule
\multicolumn{6}{l}{\textbf{TAB-Direct (Quasi-ID)}} \\
\midrule
K1 meta / direct & 0.003 & 0.932 & 1268 & 1173.6 & 0.0017 \\
K2 keywords / direct & 0.872 & 0.953 & 4 & 3.7 & 0.2717 \\
K3 case-text / direct & 0.964 & 0.964 & 1 & 1 & 1.0000 \\
K2 keywords / LLM-clues & 0.615 & 0.945 & 27 & 27.1 & 0.0369 \\
\bottomrule
\end{tabular}
\caption{TAB linkage auditing ($\alpha=0.05$). Smaller ambiguity sets indicate higher certified identifiability. \emph{direct} uses the released document as the retrieval query; \emph{LLM-clues} uses an LLM to extract a short clue string from the released document before retrieval. \emph{meta/keywords/case-text} describe which fields are available in the candidate profile database. All suites use split conformal calibration with disjoint calibration/test partitions; the full declared threat configuration per suite is summarized in Table~\ref{tab:threat-config} in Appendix~\ref{app:threat-config}.}

\label{tab:tab_results_grouped}
\end{table*}
\paragraph{Models.}
\label{sec:exp:models}
We use two open-source instruction-following LLMs as local attackers (Llama-3.1-8B-Instruct~\cite{grattafiori2024llama3} and Qwen3-8B-Instruct~\cite{yang2025qwen3}),
and one stronger cloud model (GPT-5).
Local models are used for clue extraction and attribute prediction in the sampling-based attacker.
The retrieval attacker does not require an LLM unless explicitly augmented.

\paragraph{Metrics.} We report both non-conformal ranking metrics and conformal auditing metrics.
\textit{Non-conformal baselines.}
We report Top-1/Top-$k$ accuracy, mean reciprocal rank (MRR), and mean rank of the true identity in the retrieved list.
\textit{Conformal privacy auditing.}
For each $\alpha\in\{0.01, 0.05, 0.1\}$ we report empirical coverage
$\frac{1}{N}\sum_{i=1}^N \mathbf{1}\{y_i\in C_\alpha(x_i)\}$,
the distribution of set sizes $|C_\alpha(x)|$ (median and tail percentiles),
and a leakage proxy $1/\max\{1,|C_\alpha(x)|\}$.
When the candidate pool may not contain the truth, we also report the empirical candidate-miss rate for each suite.

\subsection{Protocol and Reproducibility}
\label{sec:exp:protocol}

For each dataset and attacker configuration we use a split conformal protocol with disjoint calibration and test partitions.
Because the conformal threshold $\hat q_\alpha$ is an order statistic of the calibration scores, it depends on the calibration size; we therefore treat the split sizes as part of the declared configuration rather than as fixed parameters, state them explicitly wherever attackers are compared (e.g., $n_{\mathrm{cal}}{=}n_{\mathrm{test}}{=}100$ per TextWash category), and ablate the effect of $n_{\mathrm{cal}}$ in Section~\ref{sec:results-robustness}.
When the attack pipeline includes sampling (LLM attribute/clue extraction), sampling randomness is generated independently per example, with a declared per-document sampling budget $m$.
We calibrate a separate conformal threshold per dataset and knowledge level, enabling statistically grounded comparison across threat models
without re-tuning attack heuristics.
Attacker-to-attacker comparisons (e.g., direct retrieval vs.\ LLM-assisted clues) are made only on \emph{matched} calibration/test budgets, since certified set sizes are not comparable across different split sizes.
Whenever the release distribution or the attacker pipeline changes---including prompt, model, or web-snapshot changes---we treat the result as a \emph{new} attacker configuration and recalibrate before issuing certificates.

\section{Experimental Results}
\label{sec:results}

\begin{table*}[ht]
\centering
\caption{Retrieval-CPA on WikiBio (500 released bios; 1k candidates) under direct query vs LLM-assisted clue extraction ($\alpha=0.05$), with split conformal calibration over disjoint calibration/test partitions of the released bios; the full declared threat configuration is summarized in Table~\ref{tab:threat-config} in Appendix~\ref{app:threat-config}.}
\label{tab:wikibio}
\begin{tabular}{lrrrrrr}
\toprule
Attacker setting & Top-1 & Coverage & Median $|C_\alpha|$ & Mean $|C_\alpha|$ & $\mathbb{E}[1/|C_\alpha|]$ & $\hat q_\alpha$ \\
\midrule
Direct query & 0.362 & 0.956 & 264 & 254.98 & 0.00466 & 0.277 \\
LLM clues (Llama-3.1-8B) & 0.420 & 0.950 & 189 & 186.87 & 0.00562 & 0.2 \\
LLM clues (Qwen3-8B) & 0.380 & 0.958 & 265 & 255.69 & 0.00465 & 0.277 \\
\bottomrule
\end{tabular}
\end{table*}

% =========================
% Results (rewrite)
% =========================

\paragraph{CPA yields comparable, threat-model-specific certificates.}

Across tables and datasets, the empirical coverage is generally close to the target $1-\alpha$ (e.g., $0.95$ when $\alpha=0.05$), supporting that CPA produces a statistically meaningful certificate that is comparable across attacker configurations rather than a single, uncalibrated, attack-specific success rate reported in isolation from the declared threat model.

\paragraph{Release mechanism $\times$ auxiliary knowledge jointly determine certified identifiability.}
Table~\ref{tab:tab_results_grouped} shows a sharp transition in certified identifiability on TAB as the attacker's auxiliary knowledge increases.
With \emph{weak side information} (K1: metadata-only profiles), CPA returns near-full ambiguity sets (median $|C_\alpha|$ close to the full candidate pool), certifying substantial residual uncertainty even if the attacker occasionally ranks the true identity highly.
When the attacker is granted \emph{richer profile signals} (K2: keyword profiles; K3: case-text profiles), the ambiguity sets collapse toward singletons, certifying that the \emph{same released documents} become highly linkable under stronger auxiliary knowledge.
The quasi-identifier release variant exhibits the expected mitigation effect: compared to direct release, CPA sets expand (e.g., median $|C_\alpha|\approx 4$ under K2 keywords), indicating reduced---but still non-trivial---linkability.
Overall, TAB illustrates the central message of CPA: privacy risk is not an intrinsic property of the released text alone, but depends critically on the attacker's declared side information and tooling.

\paragraph{The miscoverage level $\alpha$ is a practical audit ``dial'' that trades certificate strength for ambiguity.}
Figure~\ref{fig:alpha_effect} (Appendix~\ref{app:robustness}) visualizes the conformal trade-off between statistical confidence and the size of the certified ambiguity set.
On TAB (direct retrieval under the quasi-identifier release), tightening the guarantee (smaller $\alpha$) yields larger ambiguity sets: at $\alpha{=}0.01$, empirical coverage is $0.998$ and the median set size is $19$ (low leakage proxy $\mathbb{E}[1/|C_\alpha|]\approx 0.052$).
Relaxing the guarantee shrinks sets and increases the leakage proxy: at $\alpha{=}0.05$, the median set size drops to $4$ ($\mathbb{E}[1/|C_\alpha|]\approx 0.233$), and at $\alpha{=}0.1$ the median reaches $1$ ($\mathbb{E}[1/|C_\alpha|]\approx 0.825$).
Once $\alpha$ becomes sufficiently large, CPA enters a near-singleton regime: ambiguity sets cannot shrink below size $1$, so the conformal certificate effectively coincides with the attacker's point prediction and empirical coverage approaches the attacker's Top-1 accuracy.
This provides a clean interpretation for auditors: \emph{smaller $\alpha$ certifies residual ambiguity more conservatively; larger $\alpha$ yields smaller certified sets at correspondingly weaker confidence.}

Blog Authorship shows the same qualitative ``dial'' behavior: with only coarse demographic-style profile fields (``W\_demo''), calibration saturates ($\hat q_\alpha=1$) and CPA returns near-full candidate sets across $\alpha$, whereas with text-enriched profiles (``W\_text'', $W{=}10$ posts/author) the median $|C_\alpha|$ drops from $9$ ($\alpha{=}0.01$) to $7$ ($\alpha{=}0.05$) to $5$ ($\alpha{=}0.1$), with coverage decreasing accordingly as the guarantee weakens.

\paragraph{LLM-assisted clue extraction can amplify linkage---but the effect is model-dependent and visible only with set-valued certificates.}
Table~\ref{tab:wikibio} shows that augmenting retrieval with LLM-generated clues can either reduce ambiguity sets (higher certified identifiability) or fail to help, depending on the attacker model.
For WikiBio, Llama-based clue extraction increases Top-1 accuracy and shrinks the median ambiguity set relative to direct retrieval, whereas Qwen-based clues behave closer to the direct baseline.
This highlights a key methodological point for LLM-era privacy auditing: different plausible attacker models can induce meaningfully different linkage risks, and CPA offers a principled way to compare them under a common calibrated output.
A similar, but more nuanced, pattern appears on TextWash under matched evaluation budgets (Table~\ref{tab:appendix:textwash-retrieval-cpa}): LLM clues reduce certified ambiguity only where external context exists, making shifts in \emph{certified residual uncertainty} visible where fixed Top-$k$ summaries would show little or no measurable change.

\paragraph{Scale and cohort effects: narrowing the candidate pool can dominate the certificate.}
When the candidate pool is large, certified ambiguity sets remain large even with enriched side information (Table~\ref{tab:blog}), indicating substantial residual uncertainty at the chosen confidence level.
In contrast, in ``motivated intruder'' scenarios where the pool is implicitly narrowed (organizational, demographic, or geographic constraints), CPA sets can collapse dramatically, certifying much higher linkability.
Realistic auditing should therefore report results under multiple candidate-pool assumptions, not only under a single global pool.
\paragraph{TextWash Retrieval-CPA under matched budgets.}
Table~\ref{tab:appendix:textwash-retrieval-cpa} reports conformal privacy auditing results on TextWash at \(\alpha=0.05\), with direct retrieval and LLM-clue attackers evaluated on the \emph{same} capped split ($n_{\mathrm{cal}}=n_{\mathrm{test}}=100$ per category), so that certified set sizes are directly comparable across attackers.
Direct retrieval yields large certified ambiguity sets (median \(|C_{\alpha}|\) of 120--169.5 depending on category), indicating that even when Top-1 accuracy is moderate, the attacker still cannot be \emph{certifiably} confident about a small identity set.
The effect of LLM-assisted clue extraction is \emph{category- and model-dependent} rather than uniform: for famous subjects, where external context exists, Llama clues shrink the certified median from 120 to 90 (Top-1 $0.470\rightarrow0.510$), whereas for fictional and semi-famous subjects---where little external knowledge is available---clue extraction leaves the certified sets essentially unchanged relative to direct retrieval.
We therefore do not claim that LLM clues uniformly dominate direct retrieval; rather, CPA makes visible \emph{where} LLM augmentation genuinely increases certified identifiability.
Empirical coverage remains close to the nominal target (\(\approx 0.95\)) across all matched runs.
\subsection{Robustness under Stressed Assumptions}
\label{sec:results-robustness}
We now stress each assumption behind the CPA guarantee empirically, making the boundary of the certificate visible rather than only theoretical (full tables and per-suite breakdowns in Appendix~\ref{app:robustness}).

\paragraph{Exchangeability and attacker/release drift.}
We calibrate under one configuration and test under another (Table~\ref{tab:shift-stress}).
Matched configurations meet the $1-\alpha=0.95$ target (TAB direct$\rightarrow$direct $0.965$; Blog W10$\rightarrow$W10 $0.940$), while shifted configurations fall below it (TAB direct$\rightarrow$quasi-ID $0.923$; Blog W10$\rightarrow$W1 $0.835$).
A CPA certificate is thus calibrated for a declared threat configuration and does \emph{not} automatically transfer under attacker or release drift---recalibration is the operational protocol.
More robust variants (e.g., Mondrian or weighted conformal prediction) are possible when the drift model is explicit, but do not provide distribution-free transfer under arbitrary prompt, model, or snapshot drift, which remains an open problem for conformal methods.

\paragraph{Candidate-pool miss and top-$K$ truncation.}
Sweeping the truncation level $K$ (Table~\ref{tab:topk-stress}) separates conformal miscoverage from candidate-pool failure.
For TextWash fiction with direct retrieval at $K=50$, the candidate-miss rate is $0.160$: coverage \emph{conditional} on candidate inclusion remains $1.000$, while unconditional coverage drops to $0.840$; at $K=200$ the miss rate vanishes and unconditional coverage recovers to $0.970$ (WikiBio+Llama behaves analogously).
Closed-world validity (Theorem~\ref{thm:cpa-valid}) applies conditional on candidate inclusion, and the observed degradation matches the multiplicative bound in Corollary~\ref{cor:open-world}.

\paragraph{Alternative score constructors at matched coverage.}
Replacing the APS mass score with rank-based, raw-probability, or calibrated top-$K$ constructors over cached LLM-clue queries (Table~\ref{tab:score-ablation}) shows that multiple constructors achieve valid coverage: on WikiBio+Llama at $\alpha=0.05$, APS attains coverage $0.950$ (median $|C_\alpha|=189$), rank-based $0.950$ (median $179$), and raw probability $0.960$ (median $212$).
CPA's calibration layer is therefore not tied to a single score; APS is a robust default rather than uniformly dominant, and rank-based scores extend CPA to attackers that expose nothing but a ranked list of candidate identities.

\paragraph{Calibration-set size.}
Varying $n_{\mathrm{cal}}\in\{50,100,150\}$ on matched TextWash direct retrieval at $\alpha=0.05$ (Table~\ref{tab:ncal-ablation}) yields average coverage $0.962/0.957/0.927$ with average median set sizes $151.2/148.2/140.3$: coverage stays near target while sets tighten with more calibration data.
Since $\hat q_\alpha$ is an order statistic of the calibration scores, certified set sizes are comparable only at matched splits; we state the split sizes explicitly for all matched comparisons and ablations.

\paragraph{Sampling budget and split stability.}
Increasing the sampling-only budget $m$ from $1$ to $10$ (Llama-3.1-8B-Instruct clue extraction on TextWash famous; Table~\ref{tab:m-ablation}) keeps coverage above target ($0.970$--$0.980$) while the median certified set size drops from $120$ to $78$: the budget affects \emph{efficiency}, not validity, once recalibrated.
Across three random calibration/test splits, coverage is stable at $0.951\pm0.037$ (median $|C_\alpha|$ $143.3\pm33.6$).
Prompt, model, or decoding changes are treated as attacker-configuration drift requiring recalibration.

\paragraph{Operational interpretation.}
The key advantage of CPA over fixed Top-$k$ reporting is that it turns a ranked list into a \emph{certificate}: for a chosen $\alpha$, the auditor can interpret $|C_\alpha(x)|$ as the (calibrated) residual uncertainty remaining after the attack.
This supports concrete release policies such as ``withhold or further redact documents with median $|C_\alpha|<k$ under attacker configuration $w$,'' and encourages tail-focused auditing of the smallest, highest-risk certified sets first.

\section{Related Work}

\label{sec:related}

\paragraph{Differential privacy for training and fine-tuning language models.} Differential privacy (DP) provides a formal stability-based guarantee for randomized algorithms \citep{dwork2006calibrating}, with DP-SGD \citep{abadi2016deep} as the standard training approach, adapted to NLP fine-tuning and generation \citep{yu2022dpfinetuning,li2022strongdp}. These works enforce a mechanism-level guarantee; in contrast, we audit \emph{task-level} re-identification risk under a declared threat model, with finite-sample guarantees derived from split conformal calibration rather than from mechanism design or calibrated noise injection.
\paragraph{Unintended memorization and data extraction attacks.} Language models can memorize and regurgitate training data or sensitive spans, enabling extraction-style attacks \citep{carlini2019secret,carlini2021extracting,carlini2022quantifying}, which motivates attack-based evaluation protocols such as membership inference and canary tests \citep{shokri2017membership,yeom2018privacy}. Our approach is compatible with these attack definitions: any attack score can serve as the underlying nonconformity signal.
\paragraph{Privacy auditing and post hoc evaluation.} A parallel literature audits trained models to estimate or falsify privacy claims, including auditing of DP pipelines \citep{jagielski2020auditing,steinke2024privacyauditing}. \citet{hu2025empiricalprivacyvariance} suggest that recovering a single scalar ``privacy number'' may be brittle across settings. Our method instead outputs a \emph{set-valued certificate} with a coverage guarantee under a declared threat family.
\paragraph{Text anonymization benchmarks and re-identification.} The Text Anonymization Benchmark (TAB) provides a corpus and evaluation framework for reducing disclosure risk in released text \citep{pilan2022tab}, and recent work argues anonymization should be assessed via de-anonymization experiments with realistic adversaries \citep{deuber2023assessing}. We complement these efforts with a distribution-free \emph{risk-control layer}: given a threat model and attacker pipeline, the reported ambiguity sets have guaranteed finite-sample coverage on held-out data.
\paragraph{Statistical disclosure control and record linkage.}
Release-time disclosure risk has a long history in statistical disclosure control (SDC) and record linkage: probabilistic record linkage scores matches against a candidate list \citep{fellegi1969theory}; disclosure-risk estimation quantifies per-record re-identification probability \citep{lambert1993measures,reiter2005estimating,hundepool2012statistical}; and $k$-anonymity bounds risk via a minimum cohort of indistinguishable records \citep{sweeney2002kanonymity}. CPA shares this candidate-list view, and its certified ambiguity set can be read as a calibrated, per-document analogue of a $k$-anonymity cohort under a \emph{modern, declared} attacker. Unlike classical SDC, which targets structured microdata and model-based population-frequency assumptions, CPA operates on unstructured text attacked by retrieval- and LLM-based pipelines and replaces model-based risk estimates with a distribution-free, finite-sample set-valued coverage guarantee. CPA does not replace classical SDC analysis; it wraps a declared modern attacker pipeline with a statistical certificate that classical disclosure-risk metrics by themselves do not provide for released text.

\paragraph{Conformal prediction and risk control.}
Conformal prediction provides finite-sample, distribution-free uncertainty quantification under exchangeability \citep{vovk2005algorithmic}, with adaptive prediction sets controlling coverage while adjusting set size to difficulty \citep{romano2020classification} and conformal \emph{risk control} extending to generic risks \citep{bates2021riskcontrol}. CPA adapts this principle to privacy: the risk is re-identification under a declared threat model and the output is certified attacker ambiguity instead of a point estimate.

\section{Conclusion}
\label{sec:conclusion}
We studied \emph{release-time} privacy risk for released text against black-box attackers that combine LLMs with auxiliary knowledge for identity linkage. Our framework, \emph{Conformal Privacy Auditing} (CPA), wraps arbitrary attacker pipelines into a \emph{finite-sample privacy certificate}, valid in both logit-access and sampling-only settings and thus applicable to open-source and proprietary API models alike. Empirically, CPA delivers calibrated coverage across benchmarks and shows that privacy risk is heterogeneous---a release that looks safe under limited side information can become highly identifiable under richer registries or stronger reasoning---so evaluation should favor calibrated, threat-model-specific certificates over fragile point estimates.

\section*{Limitations}
CPA is an \emph{auditing} framework, not a cryptographic privacy definition.
Its guarantees hold under the declared threat model: a finite candidate pool and an attacker pipeline
family specified by the auditor, together with an exchangeability assumption between calibration and
evaluation examples within each configuration.
In particular, CPA does not assume the auditor can enumerate all future attackers, and it provides
\emph{no} universal protection against unknown or unanticipated attacks---the history of linkage
attacks shows that even expert practitioners underestimate attacker ingenuity. Each certificate is
valid only for the declared candidate pool, attacker pipeline, knowledge level, and calibration/test
distribution; we therefore recommend running CPA across multiple plausible attacker families and
recalibrating whenever attacker tooling, external knowledge, or release mechanisms change. Our drift
stress tests (Section~\ref{sec:results-robustness}) quantify how coverage degrades when this protocol
is not followed.
If the candidate pool omits the true identity (open-world miss) or if deployment data drift violates
exchangeability, coverage can degrade; we therefore report candidate-miss rates together with
conditional and unconditional coverage, and recommend recalibration when the release distribution
changes (e.g., after rewriting). More robust calibration variants (e.g., Mondrian or weighted
conformal prediction under covariate shift) are possible when the drift model is explicit, but they
require additional assumptions and do not provide distribution-free transfer under arbitrary prompt
or model drift.
Finally, this work focuses empirically on identity linkage; the same calibration layer applies to any
attack that yields a scalar score or ranked candidate set, and extending CPA to membership inference
is a natural direction for future work. These limitations are not
unique to CPA but reflect the fundamental challenge of giving operationally meaningful, release-time
guarantees for natural language under evolving adversaries.

% \paragraph{Outlook.}
% A promising direction is to strengthen robustness to distribution shift and adaptive attacker behavior,
% e.g., via group-conditional calibration, shift-aware conformal methods, or online recalibration.
% Another direction is to use CPA as a decision layer for release policies: selecting redaction/rewriting
% strength to control corpus-level identifiability via conformal risk control while preserving utility.
% Overall, CPA provides a practical, statistically principled foundation for auditing identity-linkage
% risk of released text in an era where LLMs increasingly function as capable privacy adversaries.

\bibliography{example_paper}

%%%%%%%%%%%%%%%%%%%%%%%%%%%%%%%%%%%%%%%%%%%%%%%%%%%%%%%%%%%%%%%%%%%%%%%%%%%%%%%
%%%%%%%%%%%%%%%%%%%%%%%%%%%%%%%%%%%%%%%%%%%%%%%%%%%%%%%%%%%%%%%%%%%%%%%%%%%%%%%
% APPENDIX
%%%%%%%%%%%%%%%%%%%%%%%%%%%%%%%%%%%%%%%%%%%%%%%%%%%%%%%%%%%%%%%%%%%%%%%%%%%%%%%
%%%%%%%%%%%%%%%%%%%%%%%%%%%%%%%%%%%%%%%%%%%%%%%%%%%%%%%%%%%%%%%%%%%%%%%%%%%%%%%
\newpage
\appendix

\section{Appendix}
% \input{tab/tab::datasetsummary}

% \subsection{candidate profile construction clarifies the threat model and makes results interpretable across datasets.}
% Table~\ref{tab:datasets} summarizes the identity unit, document scale, and—critically—the \emph{candidate profile schema} used to instantiate attacker side information across datasets. This table is essential for interpreting CPA outputs because certified ambiguity is only meaningful relative to the candidate pool and the fields made available to the attacker. By making profile construction explicit (e.g., metadata vs keyword profiles; held-out posts vs released posts; infobox fields vs extracted keywords), we ensure that differences in ambiguity sets across datasets reflect meaningful changes in attacker knowledge rather than accidental implementation choices. This also provides a transparent checklist for future audits: practitioners can map their own auxiliary knowledge assumptions into a schema and reproduce the CPA procedure.

\subsection{CPA Algorithm}
\label{app:alg1}
The full CPA procedure is shown in Algorithm~\ref{alg:cpa}.
\begin{algorithm}[ht]
\caption{Split Conformal Privacy Auditing (CPA)}
\label{alg:cpa}
\begin{algorithmic}[1]
\REQUIRE Calibration set $\mathcal{D}_{\mathrm{cal}}=\{(X_i,Y_i)\}_{i=1}^n$ under fixed knowledge level $w$;
candidate-pool generator $\cY(\cdot,w)$; attacker pipeline $A$; miscoverage level $\alpha\in(0,1)$;
sampling budget $m$ (used only in sampling-only access).
\ENSURE Conformal threshold $\hat q_\alpha$ and ambiguity-set function $C_\alpha(\cdot,w)$.

\vspace{2pt}
\STATE \textbf{Score function (APS mass).} For any $(x,y)$ define
\[
s(x,y;w)=\sum_{\substack{y'\in \cY(x,w):\\ p_A(y'\mid x,w)\ge p_A(y\mid x,w)}} p_A(y'\mid x,w).
\]
\STATE \textbf{Calibration scores.}
\FOR{$i=1$ to $n$}
    \STATE Construct candidate pool $\cY_i \leftarrow \cY(X_i,w)$.
    \STATE Obtain a distribution $\tilde p_i(\cdot)\approx p_A(\cdot\mid X_i,w)$ over $\cY_i$:
        \IF{logits/logprobs available}
            \STATE $\tilde p_i \leftarrow p_A(\cdot\mid X_i,w)$ restricted/normalized to $\cY_i$.
        \ELSE[sampling-only access]
            \STATE Draw $\tilde Y_i^{(1)},\ldots,\tilde Y_i^{(m)} \sim A(\cdot\mid X_i,w)$ over $\cY_i$.
            \STATE Set $\tilde p_i(y)\leftarrow \frac{1}{m}\sum_{j=1}^m \mathbf{1}\{\tilde Y_i^{(j)}=y\}$ for $y\in\cY_i$.
        \ENDIF
    \STATE Compute calibration score $T_i \leftarrow \sum_{y'\in\cY_i:\,\tilde p_i(y')\ge \tilde p_i(Y_i)} \tilde p_i(y')$.
\ENDFOR

\STATE Sort $T_{(1)}\le \cdots \le T_{(n)}$ and set
\[
k \leftarrow \left\lceil (n+1)(1-\alpha)\right\rceil,\qquad \hat q_\alpha \leftarrow T_{(k)}.
\]

\STATE \textbf{Define ambiguity set for a new document $(x,w)$.}
\STATE Construct $\cY(x,w)$ and obtain $\tilde p(\cdot)\approx p_A(\cdot\mid x,w)$ as above.
\STATE Compute plug-in scores $\tilde s(y)$ for all $y\in\cY(x,w)$ as in Step 1 (with $\tilde p$ in place of $p_A$), and output
\[
C_\alpha(x,w)\leftarrow \big\{y\in \cY(x,w): \tilde s(y) \le \hat q_\alpha \big\}.
\]
\STATE \textbf{leakage proxy:} $L(x,w)\leftarrow 1/\max\{1,|C_\alpha(x,w)|\}$.
\end{algorithmic}
\end{algorithm}

\section{Dataset Statistics}
\label{app:dataset_statistics}

Table~\ref{tab:dataset_statistics} summarizes the datasets used in our privacy-auditing experiments, together with their split statistics.
\begin{table*}[ht]
\centering
\small
\resizebox{\textwidth}{!}{%
\begin{tabular}{lrrrr}
\toprule
Dataset & Unit & Total & Split / subset & Avg. words \\
\midrule
TAB (ECHR) & documents & 1{,}268 & train/dev/test = 1014/127/127 & train 1343.4; dev 874.6; test 843.4  \\
TextWash & biographies & 1{,}202 orig + 1{,}202 anon & famous 401, semifamous 400, fiction/other 401 & orig 130.5; anon 127.7  \\
Blog Authorship & authors & 19{,}320 & balanced male/female author roster & 9.74 posts/author  \\
WikiBio & biographies & 1{,}000 & one category in the prepared audit slice & orig 102.7; anon 101.9  \\
\bottomrule
\end{tabular}%
}
\caption{Dataset statistics for the release-audit experiments. For TAB we report split sizes from the benchmark JSON files; for TextWash we report file counts from the original and anonymized person-description trees; for Blog we report author-level XML statistics; and for WikiBio we report the prepared audit slice used by the canonical experiment runner in all reported WikiBio suites and ablations.}
\label{tab:dataset_statistics}
\end{table*}

The TAB benchmark contains 1014 training documents, 127 development documents, and 127 test documents. These records contain long judicial narratives, which is reflected in the substantially larger average document length relative to the other datasets.

The TextWash benchmark contains 1202 original and 1202 released descriptions, with category counts of 401 famous, 400 semi-famous, and 401 fiction/other examples. Average document length is approximately 130 words before release and 128 words after release.

The Blog authorship corpus is stored as 19320 author-level XML files. Parsing those XML files yields an average of 9.74 posts per author and a median of 6 posts per author. This corpus is substantially larger in author count than the other datasets and is used to study how CPA behaves when the attacker registry is constructed from varying numbers of posts per author.

The WikiBio experiments use the prepared audit slice emitted by the dataset-preparation script. It contains 1000 released biographies with 1000 unique identities, and the matching profile database contains 1000 profile rows. The average original and released biography lengths are 102.7 and 101.9 words respectively, reflecting the lighter editing typical of the prepared WikiBio audit slice used in our experiments.

\subsection{Implementation Details}
\label{app:impl}

\paragraph{Retrieval attacker (non-LLM and LLM-assisted).}
In Retrieval-CPA, the attacker maps a released document $x$ to a probability distribution over a candidate pool by:
(i) building candidate profile texts for each identity in the auxiliary database,
(ii) embedding them with a bag-of-words TF--IDF vectorizer, and
(iii) ranking candidates by cosine similarity to a \emph{query string} derived from $x$.
We convert similarity scores to a pseudo-posterior distribution via a temperature-scaled softmax (parameter \texttt{--softmax\_temperature}).

\emph{Direct query} (\texttt{--query\_mode direct}) uses the released text itself as the query string.
\emph{LLM clues} (\texttt{--query\_mode direct\_plus\_clues}) uses an LLM to extract a short list of missing/high-value clues from the released text (optionally augmented by web search), and appends them to the query string. This is implemented in \texttt{adversary/clue\_extract.py} and parsed via \texttt{linkage/clues.py} within \texttt{privacy\_audit\_cp}.
In both cases, we can optionally restrict to a top-$K$ candidate pool (\texttt{--top\_k\_pool}) to reduce runtime while preserving a well-defined candidate set for the split conformal calibration performed at audit time.

\paragraph{Profile attacker (LLM attribute extraction + candidate filtering).}
In Profile-CPA, the LLM adversary does not see (and is not prompted with) the full candidate list.
Instead, it predicts attributes from the released text (e.g., gender, age-range, occupation, location, or dataset-specific fields) by sampling $m$ times per attribute. The sampled attribute values induce empirical distributions $\hat p_m(\cdot\mid x)$.
For each attribute, we conformalize these distributions to produce a certified set of plausible attribute values; the final identity ambiguity set is formed by filtering the auxiliary profile database to candidates consistent with \emph{all} certified attribute sets.

\paragraph{Conformal calibration and outputs.}
Both Retrieval-CPA and Profile-CPA use split conformal calibration:
we compute a scalar nonconformity score for the true identity (or true attribute value) on a calibration split and take an order statistic threshold $\hat q_\alpha$.
At test time we output an ambiguity set $C_\alpha(x)$ consisting of candidates whose score is below $\hat q_\alpha$.
Per run we report:
(i) empirical coverage, (ii) distribution of $|C_\alpha(x)|$ (median/mean/min/max),
(iii) \emph{leakage proxy} $\mathbb{E}[1/\max\{1,|C_\alpha(x)|\}]$, and
(iv) diagnostic rates such as empty-set and candidate-miss (when applicable).
For reproducibility and debugging, the runner scripts write \texttt{jsonl} files for both calibration and test splits containing the query string, top-$K$ ranked candidates, conformal scores, and the final ambiguity set for every audited document in both conformal splits.

\paragraph{Randomness and the $U_i$ variables.}
Whenever an attacker uses sampling (e.g., $m$ LLM generations for clue extraction or attribute extraction), the randomness is represented abstractly as a per-example seed/noise variable $U_i$.
In code, $U_i$ corresponds to the randomness consumed by stochastic decoding (and any parsing post-processing). Conditional on $(X_i,W_i)$ and the fixed prompt/decoding settings, the resulting score $T_i$ is a deterministic function of $(X_i,Y_i,W_i,U_i)$, which is exactly the condition needed by the exchangeability-based conformal validity argument in Section~\ref{sec:theory}.

\paragraph{Default hyperparameters.}
Unless otherwise stated, we use:
TF--IDF with \texttt{max\_features=50{,}000}, \texttt{ngram\_max=1}, and dataset-dependent \texttt{min\_df};
softmax temperature in $[0.5,2.0]$;
LLM decoding with \texttt{temperature=0.8}, \texttt{top\_p=0.95};
and $m\in\{5,10,30\}$ samples depending on the ablation.
We use fixed random seeds for dataset splits and sampling (whenever seeding is supported by the inference backend).

\paragraph{Blog Authorship: attacker side-information via held-out posts.}
Table~\ref{tab:blog} reports Retrieval-CPA on Blog Authorship under a controlled increase in attacker side information. 
Each released document is a post whose true identity is its author. 
To model auxiliary knowledge, we build a \emph{candidate database} of author profiles by concatenating $W$ held-out posts per author ($W\in\{1,3,10\}$), and run a retrieval attacker that ranks authors by text similarity between the released post and each author profile. 
We then apply split conformal calibration at $\alpha=0.05$ to convert the attacker ranking into a certified ambiguity set $C_\alpha(x)$.
As $W$ increases, the attacker becomes stronger: Top-1 accuracy rises (0.087 $\rightarrow$ 0.209) and the conformal threshold $\hat q_\alpha$ decreases (0.948 $\rightarrow$ 0.763), yielding smaller certified ambiguity sets (median $|C_\alpha|$ drops from 222 to 176). 
Importantly, empirical coverage remains near the nominal $1-\alpha\approx0.95$, indicating that the certificates remain valid while exposing how additional auxiliary text reduces \emph{certified} uncertainty.
Despite this trend, the certified sets remain large in absolute terms, suggesting substantial residual ambiguity in a large candidate pool even when the attacker has multiple posts per author at its disposal.

\begin{table*}[ht]
\centering
\caption{Retrieval-CPA on Blog Authorship under increasing attacker side information (number of posts per author; $\alpha=0.05$), with split conformal calibration over disjoint calibration/test partitions of the released posts; the full declared threat configuration is summarized in Table~\ref{tab:threat-config} in Appendix~\ref{app:threat-config}.}
\label{tab:blog}
\begin{tabular}{lrrrrrr}
\toprule
Attacker setting & Top-1 & Coverage & Median $|C_\alpha|$ & Mean $|C_\alpha|$ & $\mathbb{E}[1/|C_\alpha|]$ & $\hat q_\alpha$ \\
\midrule
W=1 posts/author & 0.087 & 0.946 & 222 & 200.50 & 0.00566 & 0.948 \\
W=3 posts/author & 0.143 & 0.956 & 210 & 191.82 & 0.00604 & 0.903 \\
W=10 posts/author & 0.209 & 0.946 & 176 & 147.14 & 0.00677 & 0.763 \\
\bottomrule
\end{tabular}
\end{table*}

\begin{table}[t]
\centering
\caption{\textbf{TextWash Retrieval-CPA under matched budgets.}
Retrieval-CPA on TextWash under direct retrieval vs.\ LLM-assisted clue extraction (\(\alpha=0.05\)).
All attackers are evaluated on the \emph{same} capped split with \(n_{\mathrm{cal}}=n_{\mathrm{test}}=100\) per category, so certified set sizes are directly comparable across attackers; LLM rows use cached clue outputs throughout.}
\label{tab:appendix:textwash-retrieval-cpa}
\small
\setlength{\tabcolsep}{4pt}
\begin{tabular}{lrrrr}
\toprule
Attacker & Top-1 & Cov. & Med.\ \(|C_{\alpha}|\) & Mean \(|C_{\alpha}|\) \\
\midrule
\multicolumn{5}{l}{\textbf{fiction}} \\
Direct query & 0.500 & 0.970 & 169.5 & 172.44 \\
GPT-5 clues  & 0.500 & 0.970 & 169.5 & 172.44 \\
Llama clues  & 0.600 & 0.970 & 171.0 & 174.26 \\
\midrule
\multicolumn{5}{l}{\textbf{semifamous}} \\
Direct query & 0.380 & 0.910 & 155.0 & 167.12 \\
GPT-5 clues  & 0.380 & 0.910 & 155.0 & 167.12 \\
\midrule
\multicolumn{5}{l}{\textbf{famous}} \\
Direct query & 0.470 & 0.990 & 120.0 & 119.88 \\
Llama clues  & 0.510 & 0.960 &  90.0 &  89.70 \\
\bottomrule
\end{tabular}

\vspace{2pt}
\caption*{Smaller \(|C_{\alpha}|\) implies a more identifying attack (less ambiguity). Coverage should be near \(1-\alpha=0.95\) under exchangeability, up to finite-sample fluctuation.}
\end{table}

\section{Notation}
\label{app:notation}
Table~\ref{tab:notation} consolidates the symbols used in Sections~\ref{sec:setting}--\ref{sec:theory}.

\begin{table}[ht]
\centering
\small
\setlength{\tabcolsep}{3pt}
\begin{tabular}{l p{4.8cm}}
\toprule
Symbol & Meaning \\
\midrule
$x, X$ & Released document (possibly anonymized/rewritten) \\
$y, Y$ & True identity of the document subject \\
$\cU$ & Reference population of identities \\
$w, W$ & Attacker knowledge level / threat configuration (side information and tooling) \\
$\cY(x,w)$ & Finite candidate pool available to the attacker \\
$K$ & Top-$K$ truncation level of the candidate pool (engineering choice) \\
$p_A(\cdot\mid x,w)$ & Attacker's distribution (or normalized ranking) over $\cY(x,w)$ \\
$m$ & Per-document sampling budget in the sampling-only interface \\
$\hat p_m$ & Empirical attacker distribution from $m$ forced-choice samples \\
$s(x,y;w)$, $\hat s_m$ & APS mass nonconformity score (exact / plug-in) \\
$T_i$ & Scalar calibration score of example $i$ \\
$\alpha$ & Target miscoverage level \\
$\hat q_\alpha$ & Split-conformal threshold ($\lceil(n+1)(1-\alpha)\rceil$-th order statistic) \\
$n_{\mathrm{cal}}, n_{\mathrm{test}}$ & Calibration / test split sizes \\
$C_\alpha(x,w)$ & Certified ambiguity set (primary certificate) \\
$L(x,w)$ & Inverse ambiguity proxy $1/\max\{1,|C_\alpha(x,w)|\}$ (not a re-identification probability) \\
$\rho(w)$ & Candidate-miss probability $\Pr(Y\notin\cY(X,w))$ \\
$U_i$ & Attacker-internal randomness (sampling, decoding) for example $i$ \\
\bottomrule
\end{tabular}
\caption{Notation used throughout the paper.}
\label{tab:notation}
\end{table}

\section{Threat-Model Configurations per Experimental Suite}
\label{app:threat-config}
Table~\ref{tab:knowledge} summarizes the knowledge ladder used to instantiate attacker side information, and Table~\ref{tab:threat-config} maps each experimental suite to its declared threat configuration (Section~\ref{sec:threat-model}): candidate-pool construction, query mode, LLM/web augmentation, and closed- vs.\ open-world assumption. The conformal split sizes and sampling budgets of the matched and ablation suites are stated in the corresponding table captions (Tables~\ref{tab:appendix:textwash-retrieval-cpa}, \ref{tab:ncal-ablation}, and~\ref{tab:m-ablation}). Each CPA certificate is valid for exactly one row of this table; changing any entry defines a new attacker configuration that requires recalibration.
\begin{table}[t]
\centering
\small
\caption{Knowledge ladder (side information \(W\)) used to instantiate attacker pipelines in our experiments.}
\label{tab:knowledge}
\setlength{\tabcolsep}{3pt}
\begin{tabular}{l p{5.2cm}}
\toprule
Level & Attacker access / capability \\
\midrule
K1 (weak) &
Profiles contain only coarse structured fields (e.g., occupation bucket, coarse location, age range). Retrieval query uses released text directly. \\
K2 (medium) &
Profiles add additional non-identifying evidence facts (e.g., coarsened events/keywords). Retrieval uses released text; no LLM assistance. \\
K3 (strong) &
LLM-assisted clue extraction from released text produces a query expansion (optional), and/or profile is augmented (RAG/web search). \\
\bottomrule
\end{tabular}
\end{table}

\begin{table*}[ht]
\centering
\scriptsize
\setlength{\tabcolsep}{3pt}
\begin{tabular}{lllll}
\toprule
Suite & Candidate pool (size) & Query mode & LLM/web aug. & World \\
\midrule
TAB K1--K3 & case records (1268) & direct & none & closed \\
TAB K2/K3 + clues & case records (1268) & direct+clues & LLM clues & closed \\
WikiBio direct & profile rows (1000) & direct & none & closed \\
WikiBio + clues & top-$K{=}200$ of 1000 & direct+clues & LLM clues & top-$K$ truncated \\
Blog W$\in\{1,3,10\}$ & author registry & direct & none & closed \\
TextWash direct (matched) & category ($\approx$400) & direct & none & closed \\
TextWash + clues (matched) & category ($\approx$400) & direct+clues & LLM clues & closed \\
TextWash $m$-ablation & category ($\approx$400) & direct+clues & LLM clues (local HF) & closed \\
\bottomrule
\end{tabular}
\caption{Declared threat configuration for each experimental suite. ``direct'' uses the released text as the retrieval query; ``direct+clues'' appends LLM-extracted clues. The matched TextWash suites use $n_{\mathrm{cal}}{=}n_{\mathrm{test}}{=}100$ per category (Table~\ref{tab:appendix:textwash-retrieval-cpa}), and the $m$-ablation sweeps the per-document sampling budget $m\in\{1,3,5,10\}$ (Table~\ref{tab:m-ablation}).}
\label{tab:threat-config}
\end{table*}

\section{Robustness and Ablation Studies}
\label{app:robustness}
This section reports the stress tests and ablations summarized in Section~\ref{sec:results-robustness}, together with the $\alpha$-sweep visualization for the TAB quasi-identifier suite (Figure~\ref{fig:alpha_effect}).

\begin{figure*}[t]
  \centering
  \includegraphics[width=0.32\linewidth]{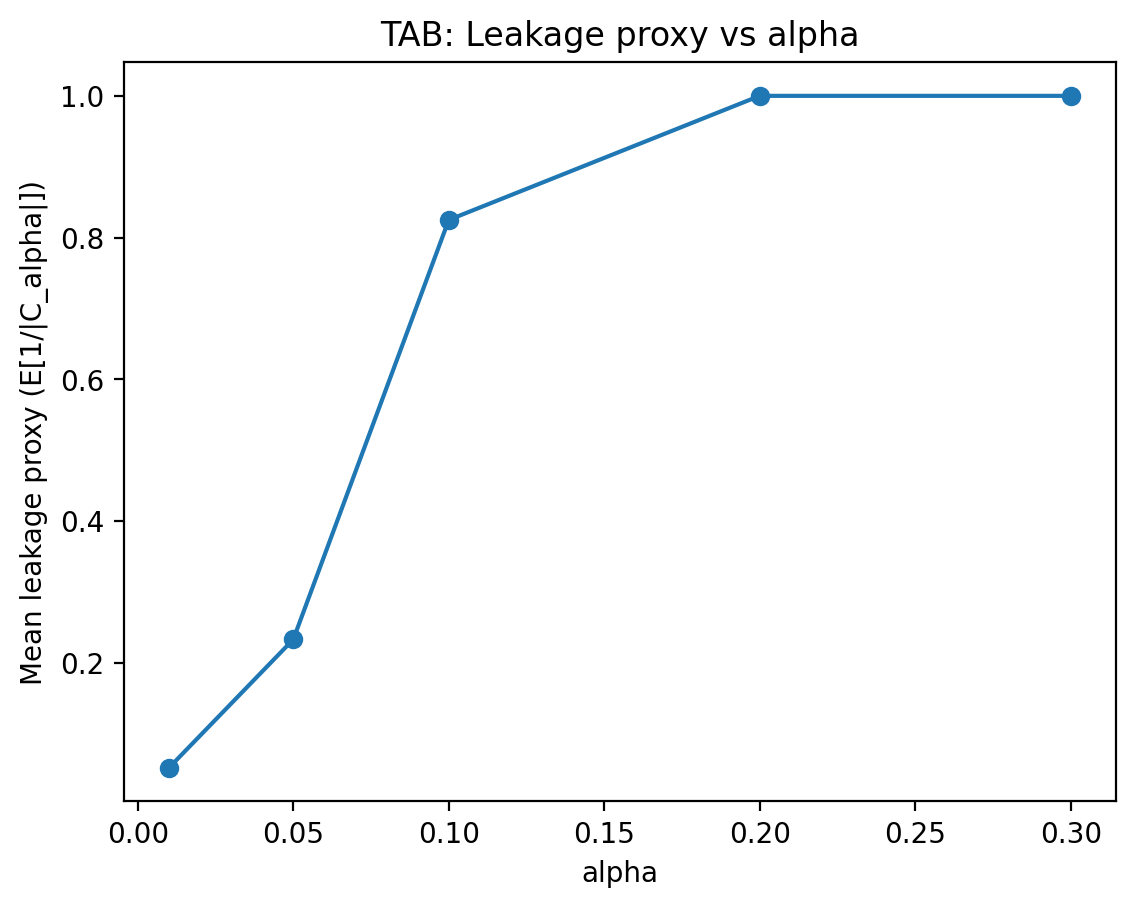}
  \includegraphics[width=0.32\linewidth]{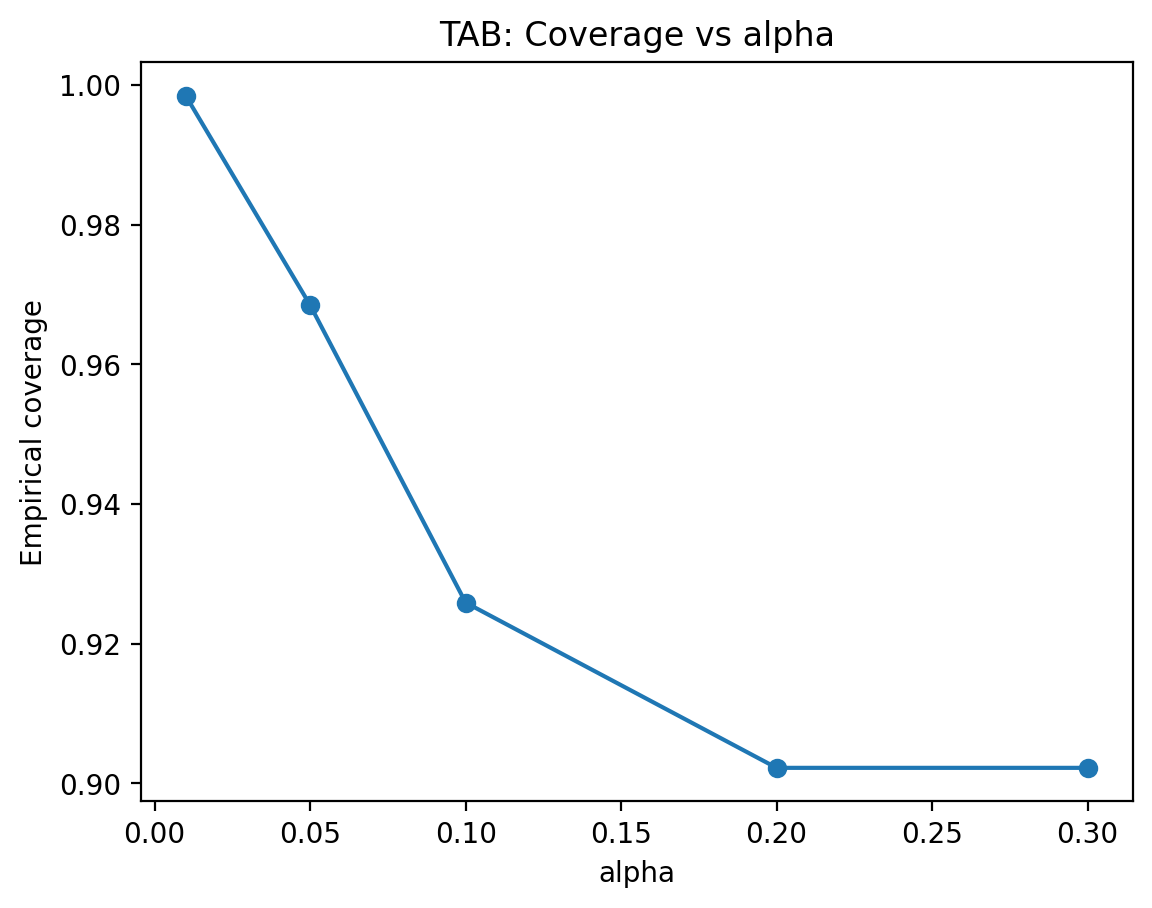}
  \includegraphics[width=0.32\linewidth]{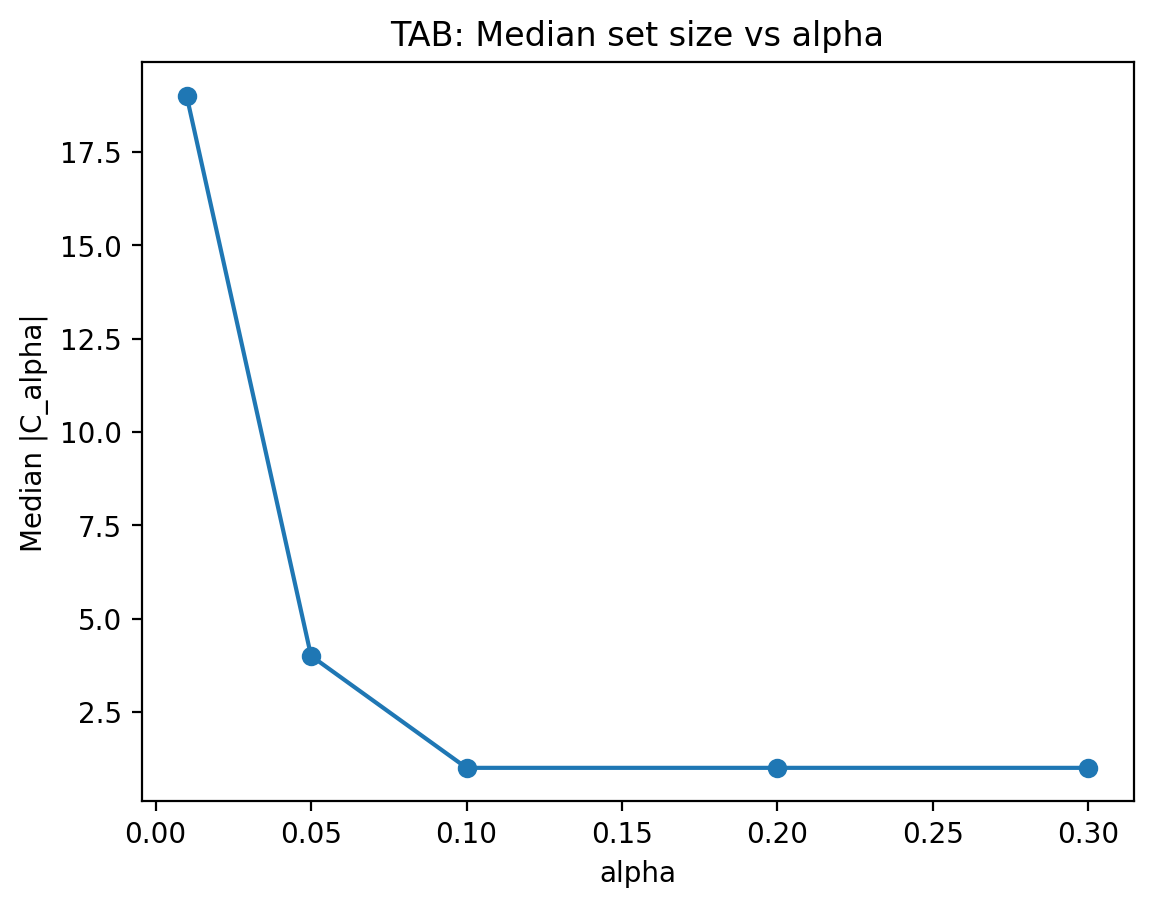}
  \caption{Effect of the miscoverage level $\alpha$ on TAB (direct retrieval, quasi-identifier release): mean leakage proxy $\mathbb{E}[1/|C_\alpha|]$ (left), empirical coverage (center), and median certified set size $|C_\alpha|$ (right). Smaller $\alpha$ certifies residual ambiguity more conservatively; larger $\alpha$ yields smaller sets with weaker statistical protection.}
  \label{fig:alpha_effect}
\end{figure*}

\subsection{Exchangeability and Attacker/Release Drift}
Table~\ref{tab:shift-stress} calibrates CPA under one configuration and evaluates it under another. For Blog, $W{=}k$ denotes an attacker registry built from $k$ known posts per author, so W10 is a stronger author-profile registry than W1. Matched configurations meet the $1-\alpha=0.95$ target, while shifted configurations fall below target, confirming that certificates do not transfer across attacker or release drift and must be recalibrated.

\begin{table}[ht]
\centering
\small
\setlength{\tabcolsep}{3pt}
\begin{tabular}{lrrr}
\toprule
Cal.\ $\rightarrow$ test & Cov. & Med.\ $|C_\alpha|$ & Mean $|C_\alpha|$ \\
\midrule
TAB direct $\rightarrow$ direct & 0.965 & 1 & 1.000 \\
TAB direct $\rightarrow$ quasi-ID & 0.923 & 1 & 1.000 \\
TAB quasi-ID $\rightarrow$ quasi-ID & 0.965 & 3 & 2.658 \\
Blog W10 $\rightarrow$ W10 & 0.940 & 176 & 175.691 \\
Blog W10 $\rightarrow$ W1 & 0.835 & 176 & 175.224 \\
\bottomrule
\end{tabular}
\caption{Exchangeability stress tests ($\alpha=0.05$). Matched calibration/test configurations achieve target coverage; shifted configurations under-cover, so recalibration is the operational protocol under suspected attacker-side or release-side drift.}
\label{tab:shift-stress}
\end{table}

\subsection{Candidate-Pool Miss and Top-$K$ Truncation}
Table~\ref{tab:topk-stress} sweeps the top-$K$ truncation level, reporting the candidate-miss rate, coverage conditional on the true identity being retained in the pool, and unconditional coverage. Conditional coverage stays at or above target whenever the true identity is retained, while unconditional coverage degrades with the miss rate, exactly as predicted by Corollary~\ref{cor:open-world}. This separates conformal miscoverage from candidate-pool construction failure: top-$K$ truncation is an engineering choice, and closed-world validity applies conditional on candidate inclusion in the truncated candidate pool.

\begin{table}[ht]
\centering
\small
\setlength{\tabcolsep}{3pt}
\begin{tabular}{lrrrrr}
\toprule
Suite & $K$ & Miss & Cond. & Uncond. & Med. \\
\midrule
TW fiction direct & 50 & 0.160 & 1.000 & 0.840 & 50 \\
TW fiction direct & 100 & 0.080 & 1.000 & 0.920 & 100 \\
TW fiction direct & 200 & 0.000 & 0.970 & 0.970 & 169.5 \\
TW famous direct & 50 & 0.110 & 1.000 & 0.890 & 50 \\
TW famous direct & 100 & 0.030 & 1.000 & 0.970 & 100 \\
WB + Llama clues & 50 & 0.118 & 1.000 & 0.882 & 50 \\
WB + Llama clues & 200 & 0.044 & 0.994 & 0.950 & 182 \\
\bottomrule
\end{tabular}
\caption{Candidate-pool / top-$K$ truncation stress tests ($\alpha=0.05$). TW = TextWash, WB = WikiBio. ``Miss'' is the empirical candidate-miss rate; ``Cond.''/``Uncond.'' are coverage conditional on the true identity being retained in the truncated pool and unconditional coverage; ``Med.'' is the median certified set size $|C_\alpha|$ over the evaluated test split.}
\label{tab:topk-stress}
\end{table}

\subsection{Alternative Score Constructors}
Table~\ref{tab:score-ablation} compares nonconformity-score constructors at matched coverage over cached LLM-clue queries (no new API calls): the APS mass score, a rank-based score, and the raw attacker probability. Multiple constructors achieve valid coverage, confirming that the CPA calibration layer is not tied to a single score; APS is a robust default rather than uniformly dominant, and the rank-based variant extends CPA to attackers that return only ranked candidate lists without scores.

\begin{table}[ht]
\centering
\small
\setlength{\tabcolsep}{3pt}
\begin{tabular}{llrrr}
\toprule
Attacker & Score & Cov. & Med. & Mean \\
\midrule
WikiBio + Llama & APS & 0.950 & 189 & 186.870 \\
WikiBio + Llama & rank & 0.950 & 179 & 179.000 \\
WikiBio + Llama & prob. & 0.960 & 212 & 209.434 \\
Blog W10 + Llama & APS & 0.950 & 8 & 8.341 \\
Blog W10 + Llama & rank & 0.962 & 9 & 9.000 \\
TAB K3 + clues & APS & 0.967 & 31 & 30.803 \\
TAB K3 + clues & prob. & 0.951 & 19 & 24.465 \\
\bottomrule
\end{tabular}
\caption{Score-constructor ablation for cached LLM-clue attackers ($\alpha=0.05$). ``Med.''/``Mean'' refer to $|C_\alpha|$. All constructors achieve valid coverage; efficiency (certified set size) varies moderately across constructors and across the different attacker pipelines.}
\label{tab:score-ablation}
\end{table}

\subsection{Calibration-Set Size}
Table~\ref{tab:ncal-ablation} varies the calibration size $n_{\mathrm{cal}}$ on matched TextWash direct retrieval ($\alpha=0.05$, $n_{\mathrm{test}}=100$ per category). Coverage remains near target while the certified sets tighten with more calibration data; per-category behavior (e.g., famous coverage $0.980\rightarrow0.940$ with median $124\rightarrow97$ as $n_{\mathrm{cal}}$ grows from 50 to 150, while fiction remains at $0.980/0.970/0.980$) shows the expected finite-sample variability of the order-statistic threshold. Because $\hat q_\alpha$ depends on $n_{\mathrm{cal}}$ through the split-conformal order statistic, the split sizes are part of the declared threat configuration and are stated with every matched comparison and ablation.

\begin{table}[ht]
\centering
\small
\begin{tabular}{lrr}
\toprule
$n_{\mathrm{cal}}$ & Avg.\ coverage & Avg.\ median $|C_\alpha|$ \\
\midrule
50 & 0.962 & 151.2 \\
100 & 0.957 & 148.2 \\
150 & 0.927 & 140.3 \\
\bottomrule
\end{tabular}
\caption{Calibration-set-size ablation on matched TextWash direct retrieval ($\alpha=0.05$, $n_{\mathrm{test}}=100$), averaged over the three TextWash categories.}
\label{tab:ncal-ablation}
\end{table}

\subsection{Sampling Budget and Split Stability}
Table~\ref{tab:m-ablation} reports a local Llama-3.1-8B-Instruct clue-extraction ablation on TextWash famous (direct-plus-clues, $n_{\mathrm{cal}}=n_{\mathrm{test}}=100$, candidate-miss rate $0$) over the per-document sampling budget $m$. Increasing $m$ from $1$ to $10$ preserves above-target coverage while reducing the median certified set size from $120$ to $78$: the sampling budget affects efficiency after recalibration, not validity.
Across three random calibration/test splits of TextWash direct retrieval (nine category-by-seed rows), coverage is $0.951\pm0.037$ with median $|C_\alpha|$ $143.3\pm33.6$ (mean $146.5\pm35.2$), indicating stability to the split randomness.
We treat prompt, model, or decoding changes as attacker-configuration drift requiring recalibration (Table~\ref{tab:shift-stress}); we do not compare or interpret certificates across such configuration changes.

\begin{table}[ht]
\centering
\small
\begin{tabular}{rrrrr}
\toprule
$m$ & Top-1 & Cov. & Med.\ $|C_\alpha|$ & Mean $|C_\alpha|$ \\
\midrule
1 & 0.490 & 0.980 & 120 & 120.020 \\
3 & 0.510 & 0.970 & 83 & 82.590 \\
5 & 0.510 & 0.970 & 82 & 82.110 \\
10 & 0.490 & 0.970 & 78 & 78.060 \\
\bottomrule
\end{tabular}
\caption{Sampling-budget ablation with Llama-3.1-8B-Instruct on TextWash famous ($\alpha=0.05$, $n_{\mathrm{cal}}=n_{\mathrm{test}}=100$, candidate-miss rate $0$). Coverage stays above target while larger budgets yield tighter certified sets at the same nominal coverage level.}
\label{tab:m-ablation}
\end{table}
%%%%%%%%%%%%%%%%%%%%%%%%%%%%%%%%%%%%%%%%%%%%%%%%%%%%%%%%%%%%%%%%%%%%%%%%%%%%%%%
%%%%%%%%%%%%%%%%%%%%%%%%%%%%%%%%%%%%%%%%%%%%%%%%%%%%%%%%%%%%%%%%%%%%%%%%%%%%%%%

\end{document}